\documentclass[final,5p,times,twocolumn,12pt]{elsarticle}
\usepackage{amsmath}
\usepackage{amssymb}
\usepackage{amsfonts}
\usepackage{graphicx}
\usepackage{multirow}
\usepackage{times}
\usepackage{caption}
\usepackage{hhline}
\usepackage[colorlinks,citecolor=blue,urlcolor=red,hyperfigures]{hyperref}
\usepackage{natbib}
\usepackage{lineno}
\usepackage{xcolor}

\biboptions{sort&compress}
\allowdisplaybreaks

\newtheorem{definition}{Definition}
\newtheorem{lemma}{Lemma}
\newtheorem{remark}{Remark}

\newtheorem{theorem}{Theorem}
\newtheorem{assumption}{Assumption}
\newenvironment{proof}[1][Proof]{\noindent\textbf{#1.} }{\ \rule{0.5em}{0.5em}}

\begin{document}

\begin{frontmatter}

\title{Event-Triggered Robust Cooperative Output Regulation for a Class of
Linear Multi-Agent Systems with an Unknown Exosystem\tnoteref{mytitlenote}}
\tnotetext[mytitlenote]{This work was supported by the National Natural Science Foundation of China (NSFC) – Excellent Young Scientists Fund (Hong Kong and Macao) under Grant 62222318.}

\author[mynewaddress]{Yangyang~Qian}
\ead{jbt4up@virginia.edu}

\author[mymainaddress]{Lu~Liu\corref{mycorrespondingauthor}}
\cortext[mycorrespondingauthor]{Corresponding author}
\ead{luliu45@cityu.edu.hk}


\address[mynewaddress]{Charles L. Brown Department of Electrical and Computer Engineering, University of Virginia, Charlottesville, VA 22904, USA}
\address[mymainaddress]{Department of Biomedical Engineering, City University of Hong Kong, Kowloon, Hong Kong}

\begin{abstract}
This paper investigates the robust cooperative output regulation problem for
a class of heterogeneous uncertain linear multi-agent systems with an
unknown exosystem via event-triggered control (ETC). By utilizing the
internal model approach and the adaptive control technique, a distributed
adaptive internal model is constructed for each agent. Then, based on this
internal model, a fully distributed ETC strategy
composed of a distributed event-triggered adaptive output feedback control
law and a distributed dynamic event-triggering mechanism is proposed, in which each
agent updates its control input at its own triggering time instants. It is
shown that under the proposed ETC strategy, the
robust cooperative output regulation problem can be solved without requiring
either the global information associated with the communication topology or
the bounds of the uncertain or unknown parameters in each agent and the
exosystem. A numerical example is provided to illustrate the effectiveness
of the proposed control strategy.
\end{abstract}


\end{frontmatter}

\section{Introduction}

Event-triggered control (ETC) was first applied to single dynamical systems
\cite%
{Astrom1999,meng2012optimal,lehmann2012comparison,Tabuada2007,girard2015dynamic}%
. In comparison with conventional periodic sampling control, ETC can often
improve the utilization efficiency of system resources such as
computational, communication or energy resources while achieving a
satisfactory control performance. Some comparison results for these two
control approaches can be found in \cite%
{Astrom1999,meng2012optimal,lehmann2012comparison}. The past decade has
witnessed a growing research interest in applying ETC to address various
cooperative control problems of multi-agent systems (MASs), for instance,
consensus, synchronization, and cooperative output regulation problems; see
\cite%
{Dimarogonas2012,Fan2013,Fan2015,yi2018dynamic,zhang2015distributed,Zhu2014,almeida2017synchronization,hao2022event,hao2023CyberEL,qian2018output,qian2018distributed,zhu2018fully,qian2021design,qian2019distributed,yang2018distributed,yang2019robust,hu2017cooperativeCyber,hu2017cooperativeTAC,hu2018robust,liu2017event,meng2018event,hao2023swithcing,qian2021cooperative,hao2022fully}
and the references therein. To address these problems, distributed ETC
strategies using only local information are utilized. Nevertheless, many
existing distributed control strategies have to depend on some global
information such as the eigenvalues of the matrices associated with the
communication topology or the number of the agents. Therefore, some effort
has been devoted to the design of fully distributed ETC strategies for MASs
\cite%
{qian2018output,qian2018distributed,zhu2018fully,qian2021design,qian2019distributed,yang2018distributed,yang2019robust,qian2021cooperative,hao2022fully,hao2023CyberEL}%
.

In recent years, the cooperative output regulation problem (CORP) has been
actively studied via ETC; see, e.g., \cite%
{qian2019distributed,yang2018distributed,yang2019robust,hu2017cooperativeCyber,hu2017cooperativeTAC,hu2018robust,liu2017event,meng2018event,qian2021cooperative,hao2022fully,hao2023swithcing}%
. This problem is more general than the problems of leader-following
consensus and synchronization in that it can handle not only trajectory
tracking but also disturbance rejection for MASs. In the literature, fully
distributed ETC strategies have been designed to address the CORP; see \cite%
{qian2019distributed,yang2018distributed,yang2019robust,qian2021cooperative,hao2022fully}%
. In \cite{qian2019distributed} and \cite{yang2018distributed}, the CORP of
heterogeneous linear MASs under fixed or switching topologies was addressed
via the feedforward design approach. In \cite{yang2019robust}, the CORP of
heterogeneous linear MASs with parameter uncertainties was studied via the
internal model approach. The ETC strategies in \cite%
{qian2019distributed,yang2018distributed,yang2019robust} reply on the
assumption that the system matrix of the exosystem is known to all agents.
To relax this assumption, event-based adaptive distributed observers were
designed to solve the CORP of heterogeneous linear MASs under fixed \cite%
{qian2021cooperative} or switching \cite{hao2022fully} topologies, under the assumption that
only a subset of agents connected to the exosystem can access the system
matrix of the exosystem. It is noted that by introducing some adaptive
parameters into the ETC law and the event-triggering mechanism, no global
information is needed in \cite%
{qian2019distributed,yang2018distributed,yang2019robust,qian2021cooperative}.

In the CORP, the reference signals and external disturbances are assumed to
be generated by an exosystem. Such an exosystem is required to be exactly
known in the existing studies \cite%
{qian2019distributed,yang2018distributed,yang2019robust,hu2017cooperativeCyber,hu2017cooperativeTAC,hu2018robust,liu2017event,meng2018event,qian2021cooperative,hao2022fully}%
. However, in many practical scenarios, the model of the exosystem may not
be exactly known. For example, the frequency of a sinusoidal external
disturbance could be unknown. This type of disturbances can be generated by
an exosystem with unknown parameters, called an unknown exosystem.
Additionally, there often exist uncertainties in agent dynamics due to
various reasons such as inaccurate measurement and environmental
perturbation. Thus, it is more practical to study the CORP of uncertain MASs
with an unknown exosystem under ETC. For the special case of a single agent,
ETC strategies have been developed in \cite{qian2021event,liu2020dynamic} to
address the robust output regulation problem in the presence of an unknown
exosystem. In particular, based on an adaptive internal model, an ETC
strategy was proposed in \cite{qian2021event} for uncertain linear systems
and in \cite{liu2020dynamic} for uncertain nonlinear systems.

This paper investigates the robust CORP with an unknown exosystem via fully
distributed ETC strategies for a class of heterogeneous uncertain linear
MASs with arbitrarily large parameter uncertainties in agent dynamics. It is
noted that the design methods in \cite{yang2019robust} and \cite%
{hu2018robust} are not applicable for solving the considered problem because
they require the parameter uncertainties to be sufficiently small. Instead,
a distributed adaptive internal model is constructed for each agent by
utilizing the internal model approach and the adaptive control technique.
Then, based on this internal model, a fully distributed ETC strategy
composed of a distributed event-triggered adaptive output feedback control
law and a distributed event-triggering mechanism is proposed, in which each
agent updates its control input at its own triggering time instants. It is
shown that under the proposed ETC strategy, the robust CORP can be solved.
Moreover, by establishing an explicit positive lower bound of
inter-event times, it is proven that the exclusion of the Zeno behavior is
ensured for each agent.

Compared with the existing relevant studies, the contributions of this paper
are summarized as follows. First, this paper presents a fully distributed
ETC strategy that can handle uncertain MASs with an unknown exosystem and
reduce the number of controller updates. Second, the proposed ETC strategy
does not require either the global information associated with the
communication topology or the bounds of the uncertain or unknown parameters
in each agent and the exosystem, which is in contrast to the control
strategies in \cite{hu2018robust} and \cite{liu2017event}. Third, the
parameter uncertainties in this paper are allowed to be arbitrarily large
while the parameter uncertainties in \cite{yang2019robust} and \cite%
{hu2018robust} are required to be sufficiently small. It should be pointed
out that although the proposed ETC strategy shares the aforementioned
advantages, it requires continuous communication between neighboring agents
when using inter-agent communication, while the ETC strategies in \cite%
{qian2019distributed,yang2018distributed,yang2019robust,qian2021cooperative} only require
intermittent communication.

\textit{Notation:} Let $\mathbb{R}$ and $\mathbb{N}$ be the sets of real
numbers and nonnegative integers, respectively. The symbol $\left\Vert \cdot
\right\Vert $ stands for the Euclidean norm for vectors or the induced $2$%
-norm for matrices. The transpose of matrix $A$ is denoted as $A^{\mathrm{T}%
} $. The Kronecker product of matrices $A $ and $B$ is denoted as $A\otimes
B $. Define $\mathrm{col}(A_{1}$, $\cdots $, $A_{N})=[A_{1}^{\mathrm{T}}$, $%
\cdots $, $A_{N}^{\mathrm{T}}]^{\mathrm{T}} $, with $A_{i}\in \mathbb{R}%
^{n_{i}\times m}$, $i=1$, $\cdots $, $N$. Let $\mathrm{diag}\{a_{1}\text{, }%
\cdots \text{, }a_{N}\}$ be a diagonal matrix with $a_{i}\in \mathbb{R}$ as
the $i$th diagonal element, and $\mathrm{blockdiag}\{A_{1}\text{, }\cdots
\text{, }A_{N}\}$ be a block diagonal matrix with $A_{i}\in \mathbb{R}%
^{n_{i}\times n_{i}}$ as the $i$th diagonal element. $\mathbf{1}_{N}$ is an $%
N$-dimensional column vector with all elements being $1$. $I_{N}$ is an
identity matrix of dimension $N\times N$. If the left-hand and right-hand
derivatives of a scalar function $\alpha(t)$ at time $t$ exist but are not
equal, then $\dot{\alpha}(t)$ or $(d/dt)\alpha (t)$ represents the
right-hand derivative of $\alpha (t)$.

\section{Problem Formulation and Preliminaries\label{Sec2}}

\subsection{Problem Formulation}

Consider a class of heterogeneous uncertain linear MASs described by%
\begin{align}
\dot{z}_{i}& =A_{1i}(w)z_{i}+A_{2i}(w)\xi _{i}+E_{0i}(w)v\text{,}  \notag \\
\dot{\xi}_{i}& =A_{3i}(w)z_{i}+A_{4i}(w)\xi _{i}+E_{1i}(w)v+b_{i}(w)u_{i}%
\text{,}  \notag \\
y_{i}& =\xi _{i}\text{, } i=1\text{, }\cdots \text{, }N\text{,}  \label{1}
\end{align}%
where $\mathrm{col}(z_{i}$, $\xi _{i})$ $\in \mathbb{R}^{n_{i}-1}\times
\mathbb{R}$, $u_{i}\in \mathbb{R}$, and $y_{i}\in \mathbb{R}$ are the state,
control input, and measurement output of agent $i$, respectively, $w\in
\mathbb{R}^{n_{w}}$ is an uncertain parameter vector belonging to the
unknown compact subset $\mathbb{W\subset R}^{n_{w}}$ with $w=0$ as its
nominal value, $b_{i}(w)>0$ for all $w$, and $v\in \mathbb{R}^{q}$ is the
exogenous signal representing the reference input or the disturbance or
both. The signal $v$ is generated by an unknown exosystem of the form%
\begin{equation}
\dot{v}=S(\sigma )v\text{, }y_{0}=F(w)v\text{,}  \label{2}
\end{equation}%
where $S(\sigma)\in \mathbb{R}^{q\times q}$, $y_{0}\in \mathbb{R}$ is the
reference output, and $\sigma \in \mathbb{R}^{n_{\sigma }}$ is an unknown
parameter vector belonging to the unknown compact subset $\mathbb{S\subset R}%
^{n_{\sigma}}$. It can be seen that exosystem (\ref{2}) contains the unknown
parameter vector $\sigma$. In this sense, exosystem (\ref{2}) is said to be
unknown. For the systems (\ref{1}) and (\ref{2}), $A_{1i}(w)\in \mathbb{R}%
^{(n_{i}-1)\times (n_{i}-1)}$, $A_{2i}(w)\in \mathbb{R}^{(n_{i}-1)\times 1}$%
, $A_{3i}(w)\in \mathbb{R}^{1\times (n_{i}-1)}$, $A_{4i}(w)\in \mathbb{R}$, $%
E_{0i}(w)\in \mathbb{R}^{(n_{i}-1)\times q}$, $E_{1i}(w)\in \mathbb{R}%
^{1\times q}$, and $F(w)\in \mathbb{R}^{1\times q}$ are all continuous in $w$%
. Note that each agent in system (\ref{1}) is in the normal form with unity
relative degree (see \cite[Section 7.5]{chen2015stabilization}), and only a
subset of agents can have access to the output $y_{0}$.

Let $N$ agents be indexed by $1$, $\cdots $, $N$, and the exosystem be
indexed by $0$. Then, the communication network among the nodes $0$, $1$, $%
\cdots $, $N$ is modeled by a directed graph\footnote{%
See \cite{godsil2013algebraic} for some terminologies on directed graphs.} $%
\mathcal{\bar{G}}=(\mathcal{\bar{N}}$, $\mathcal{\bar{E}})$, where $\mathcal{%
\bar{N}}=\{0$, $1$, $\cdots $, $N\}$ is the node set and $\mathcal{\bar{E}=%
\bar{N}\times \bar{N}}$ is the edge set. If $(j,i)\in \mathcal{\bar{E}}$,
for $i$, $j\in \mathcal{\bar{N}}$, then node $i$ can receive information
from node $j$. In this case, node $j$ is called an in-neighbor of node $i$,
and node $i$ is called an output-neighbor of node $j$. Note that for $i=1$, $%
\cdots $, $N$, $(0,i)\in \mathcal{\bar{E}}$, if and only if agent $i$ can
access the signal $y_{0}$, and $(i,0)\notin \mathcal{\bar{E}}$, since the
exosystem has no control input. Let the adjacency matrix of $\mathcal{\bar{G}%
}$ be $\mathcal{\bar{A}}=[a_{ij}]\in \mathbb{R}^{(N+1)\times (N+1)}$, $i$, $%
j\in \mathcal{\bar{N}}$, where $a_{ii}=0$; $a_{ij}=1$, $j\neq i$, if $%
(j,i)\in \mathcal{\bar{E}}$; otherwise $a_{ij}=0$. Furthermore, define a
subgraph $\mathcal{G}=(\mathcal{N}$, $\mathcal{E})$ of $\mathcal{\bar{G}}$,
where $\mathcal{N}=\{1 $, $\cdots $, $N\}$, and $\mathcal{E}=\mathcal{N}%
\times \mathcal{N}$ is obtained from $\mathcal{\bar{E}}$ by removing all the
edges between node $0$ and the nodes in $\mathcal{N}$.

The problem under consideration is defined as follows.

\begin{definition}[Robust CORP]
\label{CRORP}Consider the heterogeneous uncertain linear MAS (\ref{1}) with
the unknown exosystem (\ref{2}) under a directed communication graph $%
\mathcal{\bar{G}}$. Design a distributed ETC strategy such that for any $%
w\in \mathbb{W\subset R}^{n_{w}}$ and $\sigma \in \mathbb{S\subset R}%
^{n_{\sigma }}$, the following two properties are satisfied:

\begin{enumerate}
\item the solution of the resulting closed-loop system is bounded for all $%
t\geq 0$; and

\item for any initial condition $z_{i}(0)$, $\xi _{i}(0)$, $i=1$, $\cdots $,
$N$, and $v(0)$, the tracking error $e_{i}\triangleq y_{i}-y_{0}$ satisfies $%
\lim_{t\rightarrow \infty }\left\vert e_{i}(t)\right\vert \leq \varepsilon
_{i}$, where $\varepsilon _{i}>0$ is any specified level of accuracy, for all $i=1$, $%
\cdots $, $N$.
\end{enumerate}
\end{definition}

In this paper, the main aim is to design a fully distributed ETC strategy
without requiring any global information such that the robust CORP for the MAS (\ref{1}) and exosystem (\ref{2}) can
be solved. To proceed, the following assumptions are introduced.

\begin{assumption}
\label{Asmp1}All the eigenvalues of $S$ are semi-simple with zero real parts.
\end{assumption}

\begin{assumption}
\label{Asmp2}For $i=1$, $\cdots $, $N$, agent $i$ is a minimum phase system
for all $w\in \mathbb{W}$.
\end{assumption}

\begin{assumption}
\label{Asmp3}The graph $\mathcal{\bar{G}}$ has a directed spanning tree with
node $0$ as the root, and its subgraph $\mathcal{G}$ is undirected.
\end{assumption}

\begin{remark}
\label{Rem1}It is noted that Assumptions \ref{Asmp1} and \ref{Asmp2} have
also been used in \cite{liu2017event}. Under Assumption \ref{Asmp1}, the
exogenous signal $v(t) $ is bounded for $t\in \lbrack 0,\infty )$. Under
Assumption \ref{Asmp2}, the matrices $A_{1i}(w)$, $i=1$, $\cdots $, $N$, are
Hurwitz for all $w\in \mathbb{W}$. Let $H=[h_{ij}]\in \mathbb{R}^{N\times N}$%
, $i$, $j\in \mathcal{N}$, with $h_{ii}=$ $\sum_{j=0}^{N}a_{ij}$ and $%
h_{ij}=-a_{ij}$ for $i\neq j$. Then, under Assumption \ref{Asmp3}, $H$ is
symmetric. Moreover, by Lemma 1 of \cite{su2012cooperativeTAC}, all the
eigenvalues of $H$ have positive real parts. Thus, $H$ is a positive
definite matrix.
\end{remark}

\begin{remark}
\label{Rem2}In Assumption \ref{Asmp3}, the graph $\mathcal{G}$ is required
to be undirected, instead of being directed as in \cite{li2016distributed}.
The difficulties of solving the considered problem in the case of a directed graph $%
\mathcal{G}$ lie in the constraints arising from the ETC and the asymmetry
of the Laplacian matrix associated with the directed graph $\mathcal{G}$.
\end{remark}

\subsection{Preliminaries}

As shown in \cite{su2014cooperative}, the robust CORP of the MAS (\ref{1}) and exosystem (\ref{2}) can be
converted into a robust cooperative stabilization problem of an augmented
system. In order to facilitate this conversion, some notations are defined
as follows.

Let $\lambda ^{l}+a_{l-1}(\sigma )\lambda ^{l-1}+\cdots +a_{1}(\sigma
)\lambda +a_{0}(\sigma )$ be the minimal polynomial of $S(\sigma )$, where $%
a_{0}(\sigma )$, $a_{1}(\sigma )$, $\cdots $, $a_{l-1}(\sigma )$ are some
real numbers. Denote
\begin{align*}
\Phi (\sigma ) &=\left[
\begin{array}{cccc}
0 & 1 & \cdots & 0 \\
\vdots & \vdots & \ddots & \vdots \\
0 & 0 & \cdots & 1 \\
-a_{0}(\sigma ) & -a_{1}(\sigma ) & \cdots & -a_{l-1}(\sigma )%
\end{array}%
\right] \text{,} \\
\Gamma &=\left[
\begin{array}{cccc}
1 & 0 & \cdots & 0%
\end{array}%
\right] _{1\times l}\text{.}
\end{align*}%
Since for $i=1$, $\cdots $, $N$, the matrices $S(\sigma)$ and $A_{1i}(w)$
have no common eigenvalues, the Sylvester equation
\begin{equation}
\Pi _{i}(w,\sigma )S(\sigma )=A_{1i}(w)\Pi _{i}(w,\sigma
)+A_{2i}(w)F(w)+E_{0i}(w)  \label{4}
\end{equation}%
has a unique solution $\Pi _{i}(w,\sigma)\in \mathbb{R}^{(n_{i}-1)\times q}$%
. Let
\begin{align*}
U_{i}(w)&=-\frac{1}{b_{i}(w)}(A_{3i}(w)\Pi _{i}(w,\sigma )+A_{4i}(w)F(w)
\notag \\
&\quad +E_{1i}(w)-F(w)S(\sigma )) \text{,} \\
\Upsilon _{i}(w,\sigma)&=\mathrm{col}(U_{i}(w,\sigma )\text{, }%
U_{i}(w,\sigma )S(\sigma )\text{, }\cdots \text{, }  \notag \\
&\quad U_{i}(w,\sigma)S^{l-1}(\sigma ))\text{.}
\end{align*}
Then, for all $\mathrm{col}(w,\sigma )\in \mathbb{W\times S}$,
\begin{align}
U_{i}(w,\sigma ) &=\Gamma \Upsilon _{i}(w,\sigma )\text{,}  \label{5} \\
\Upsilon _{i}(w,\sigma )S(\sigma ) &=\Phi (\sigma )\Upsilon _{i}(w,\sigma )%
\text{.}  \label{6}
\end{align}

Let $M\in \mathbb{R}^{l\times l}$ and $Q\in \mathbb{R}^{l\times 1}$ be any
controllable pair with $l$ being the degree of the minimal polynomial of $S$%
, where $M$ is Hurwitz. Since $(\Gamma ,\Phi )$ is observable and the
matrices $\Phi $ and $M$ have no common eigenvalues, the Sylvester equation
\begin{equation}
T(\sigma )\Phi (\sigma )-MT(\sigma )=Q\Gamma  \label{7}
\end{equation}%
has a unique solution $T(\sigma )\in \mathbb{R}^{l\times l}$ which is
nonsingular (see \cite[Proposition A.2]{huang2004nonlinear}). Let $\Psi
_{\sigma }=\Gamma T^{-1}(\sigma )$ and $\bar{\Upsilon}_{i}(w,\sigma
)=T(\sigma )\Upsilon _{i}(w,\sigma )$. Then, it follows from (\ref{5}) that%
\begin{equation}
U_{i}(w,\sigma )=\Psi _{\sigma }\bar{\Upsilon}_{i}(w,\sigma )\text{.}
\label{8}
\end{equation}%
Since the parameter vector $\sigma $ is unknown, $\Psi _{\sigma }$ is also
unknown. To deal with the unknown $\Psi _{\sigma }$, an estimator will be
constructed to estimate $\Psi _{\sigma }$ for each individual agent. The
specific form of the designed estimator is shown in the next section.

As shown in \cite{su2014cooperative}, under Assumptions \ref{Asmp1}-\ref%
{Asmp2}, an internal model can be constructed as%
\begin{equation}
\dot{\eta}_{i}=M\eta _{i}+Qu_{i}\text{, }i=1,\cdots ,N\text{.}  \label{9}
\end{equation}%
Note that the composition of the MAS (\ref{1}) and the internal model (\ref%
{9}) is called the augmented system. Define the following coordinate and
input transformation:
\begin{align}
\bar{z}_{i}& =z_{i}-\Pi _{i}(w,\sigma )v\text{,}  \notag \\
\bar{\xi}_{i}& =\xi _{i}-F(w)v\text{,}  \notag \\
\bar{\eta}_{i}& =\eta _{i}-\bar{\Upsilon}_{i}(w,\sigma )v-\frac{1}{b_{i}(w)}Q%
\bar{\xi}_{i}\text{,}  \notag \\
\tilde{\Psi}_{i}& =\hat{\Psi}_{i}-\Psi _{\sigma }\text{, }  \notag \\
\bar{u}_{i}& =u_{i}-\hat{\Psi}_{i}(t_{k}^{i})\eta _{i}(t_{k}^{i})\text{, }i=1%
\text{, }\cdots \text{, }N\text{,}  \label{10}
\end{align}%
where for agent $i$, $\hat{\Psi}_{i}$ is the estimate of $\Psi _{\sigma }$
and $t_{k}^{i}$ is the $k$th triggering time instant with $k\in \mathbb{N}$.
Then, performing the transformation (\ref{10}) on the MAS (\ref{1}) and the
internal model (\ref{9}) yields the following augmented system:%
\begin{align}
\dot{\bar{z}}_{i}& =A_{1i}(w)\bar{z}_{i}+A_{2i}(w)\bar{\xi}_{i}\text{,}
\notag \\
\dot{\bar{\xi}}_{i}& =A_{3i}(w)\bar{z}_{i}+A_{5i}(w)\bar{\xi}_{i}+b_{i}(w)%
\bar{u}_{i}+b_{i}(w)\Psi _{\sigma }\bar{\eta}_{i}  \notag \\
& \quad +b_{i}(w)\tilde{\Psi}_{i}\eta _{i}+b_{i}(w)\left( \hat{\Psi}%
_{i}(t_{k}^{i})\eta _{i}(t_{k}^{i})-\hat{\Psi}_{i}\eta _{i}\right) \text{,}
\notag \\
\dot{\bar{\eta}}_{i}& =M\bar{\eta}_{i}+A_{6i}(w)\bar{z}_{i}+A_{7i}(w)\bar{\xi%
}_{i}\text{, } i=1\text{, }\cdots \text{, }N\text{,}  \label{11}
\end{align}%
where $A_{5i}(w)\triangleq A_{4i}(w)+\Psi _{\sigma }Q$, $A_{6i}(w)\triangleq
-(1/b_{i}(w))QA_{3i}(w)$, $A_{7i}(w) \triangleq (1/b_{i}(w))$ $(MQ+Q\Psi
_{\sigma }Q-QA_{5i}(w))$, and the equations (\ref{4})-(\ref{8}) have been
used. It can be seen that the tracking error $e_{i}$ is given by $e_{i} =%
\bar{\xi}_{i}$. To solve the considered problem, it suffices to solve a
robust cooperative stabilization problem of system (\ref{11}), that is, to
design a fully distributed ETC strategy, such that for any initial condition
and for any $w\in \mathbb{W}$ and $\sigma \in \mathbb{S}$, the solution of
the resulting closed-loop system is bounded for all $t\geq 0$, and the
tracking errors $e_{i}$, $i=1$, $\cdots $, $N$ converge to the origin with any specified level of accuracy.

\section{Main Results\label{Sec3}}

In this section, a fully distributed ETC strategy will be proposed to solve
the robust CORP for the uncertain linear MAS (\ref{1}) and the unknown
exosystem (\ref{2}). To this end, the robust cooperative stabilization
problem of system (\ref{11}) is first solved as follows.

Define
\begin{equation}
e_{vi}=\sum_{j=0}^{N}a_{ij}\left( y_{i}-y_{j}\right)
=\sum_{j=0}^{N}a_{ij}\left( e_{i}-e_{j}\right) \text{,}  \label{12}
\end{equation}%
with $e_{0}=0$, for $i=1$, $\cdots $, $N$. Then, for system (\ref{11}), a
fully distributed event-triggered adaptive output feedback control law is
designed as follows:
\begin{align}
\bar{u}_{i}(t)& =-K_{i}(t_{k}^{i})e_{vi}(t_{k}^{i})\text{,}  \notag \\
\dot{\hat{\Psi}}_{i}^{\mathrm{T}}(t)& =-\gamma _{i}e_{vi}(t)\eta _{i}(t)%
\text{,}  \notag \\
\dot{K}_{i}(t)& =\delta _{i}e_{vi}^{2}(t)\text{, }i=1\text{, }\cdots \text{,
}N\text{,}  \label{13}
\end{align}%
for $t\in \lbrack t_{k}^{i}$, $t_{k+1}^{i})$, $k\in \mathbb{N}$, where $%
K_{i}(t)$ is the controller gain with the initial value $K_{i}(0)>0$, $\hat{%
\Psi}_{i}(t)$ is the estimate of the unknown parameter vector $\Psi _{\sigma
}$, and $\gamma _{i}>0$ and $\delta _{i}>0$ are the adaptation gains.

\begin{remark}
\label{Rem3}The gain $K_{i}(t)$ is dynamically updated by the adaptive law $%
\dot{K}_{i}(t)=\delta _{i}e_{vi}^{2}(t)$, and hence is called an adaptive
gain. By introducing the adaptive gain $K_{i}(t)$, the proposed distributed
ETC law (\ref{13}) has the following two features. One is that the control
law (\ref{13}) is independent of the bounds of the parameters $w$ and $%
\sigma $, and hence the bounds of the compact subsets $\mathbb{W}$ and $%
\mathbb{S}$ are allowed to be unknown. The other is that the control law (%
\ref{13}) is fully distributed in the sense that the design parameters in (%
\ref{13}) do not rely on any global information such as the eigenvalues of $%
H $ and the number $N$.
\end{remark}

To determine when agent $i$ updates its control input $\bar{u}_{i}(t)$, a
distributed event-triggering mechanism will be developed. Define two
measurement errors as
\begin{align}
\zeta _{1i}(t)& =K_{i}(t_{k}^{i})e_{vi}(t_{k}^{i})-K_{i}(t)e_{vi}(t)\text{,}
\label{14} \\
\zeta _{2i}(t)& =\hat{\Psi}_{i}(t_{k}^{i})\eta _{i}(t_{k}^{i})-\hat{\Psi}%
_{i}(t)\eta _{i}(t)\text{,}  \label{15}
\end{align}%
for $t\in \lbrack t_{k}^{i}$, $t_{k+1}^{i})$, and further define an
event-triggering function as
\begin{equation}
f_{i}(\zeta _{1i},\zeta _{2i},e_{vi},t)=\zeta _{1i}^{2}(t)+\zeta
_{2i}^{2}(t)-\kappa _{i}e_{vi}^{2}(t)-\beta _{i}\text{,}  \label{16}
\end{equation}%
where $0<\kappa _{i}<1$ and $\beta _{i}>0$ are design parameters. To
facilitate the design of the dynamic triggering mechanism, define a dynamic
trigger variable $h_{i}(t)$ governed by
\begin{equation}
\dot{h}_{i}(t)=-\alpha _{i}h_{i}(t)-f_{i}(\zeta _{1i},\zeta _{2i},e_{vi},t)%
\text{,}  \label{dynamic}
\end{equation}%
where $h_{i}(0)>0$ and $\alpha _{i}>0$ is a scalar gain.

Let $\{t_{k}^{i}:k\in \mathbb{N}\}$ be the sequence of triggering time
instants of agent $i$, for $i=1$, $\cdots $, $N$. Then, to determine the
sequence $\{t_{k}^{i}:k\in \mathbb{N}\}$, a fully distributed dynamic
event-triggering mechanism is designed in the form of%
\begin{equation}
t_{k+1}^{i}=\inf \left\{ t>t_{k}^{i}\mid f_{i}(\zeta _{1i},\zeta
_{2i},e_{vi},t)\geq h_{i}(t)\right\} \text{.}  \label{17}
\end{equation}%
The triggering mechanism (\ref{17}) is called a dynamic one because of the existence of the trigger variable $h_{i}(t)$. Note that $f_{i}(\zeta _{1i},\zeta
_{2i},e_{vi},t)\geq h_{i}(t)$ is called the triggering condition. This
condition means that the total measurement error $\zeta
_{1i}^{2}(t)+\zeta_{2i}^{2}(t)$ is equal to or larger than the triggering
threshold $\kappa _{i}e_{vi}^{2}(t)+\beta _{i}+h_{i}(t)$. Once the
triggering condition $f_{i}(\zeta _{1i},\zeta _{2i},e_{vi},t)\geq h_{i}(t)$
holds, agent $i$ needs to sample the signals $K_{i}$, $e_{vi}$, $\hat{\Psi}%
_{i}$, and $\eta _{i}$, and update the control input $\bar{u}_{i}(t)$ at the
same time.

By applying the control law (\ref{13}) to the augmented system (\ref{11}),
the resulting closed-loop system is given by%
\begin{align}
\dot{\bar{z}}_{i}& =A_{1i}(w)\bar{z}_{i}+A_{2i}(w)\bar{\xi}_{i}\text{,}
\notag \\
\dot{\bar{\eta}}_{i}& =M\bar{\eta}_{i}+A_{6i}(w)\bar{z}_{i}+A_{7i}(w)\bar{\xi%
}_{i}\text{,}  \notag \\
\dot{\bar{\xi}}_{i}& =A_{3i}(w)\bar{z}_{i}+A_{5i}(w)\bar{\xi}%
_{i}-b_{i}(w)K_{i}e_{vi}  \notag \\
& \quad +b_{i}(w)\Psi _{\sigma }\bar{\eta}_{i}+b_{i}(w)\tilde{\Psi}_{i}\eta
_{i}+b_{i}(w)\zeta _{2i}-b_{i}(w)\zeta _{1i}\text{,}  \notag \\
\dot{\hat{\Psi}}_{i}^{\mathrm{T}}& =-\gamma _{i}e_{vi}\eta _{i}\text{,}
\notag \\
\dot{K}_{i}& =\delta _{i}e_{vi}^{2}\text{, }  \notag \\
\dot{h}_{i}& =-\alpha _{i}h_{i}-f_{i}(\zeta _{1i},\zeta _{2i},e_{vi},t)\text{%
, }i=1\text{, }\cdots \text{, }N\text{,}  \label{18}
\end{align}%
where the definitions (\ref{14}) and (\ref{15}) have been used. Let $e=%
\mathrm{col}(e_{1}$, $\cdots $, $e_{N})$, $e_{v}=\mathrm{col}(e_{v1}$, $%
\cdots $, $e_{vN})$, $\bar{z}=\mathrm{col}(\bar{z}_{1}$, $\cdots $, $\bar{z}%
_{N})$, $\bar{\eta}=\mathrm{col}(\bar{\eta}_{1}$, $\cdots $, $\bar{\eta}%
_{N}) $, $\bar{\xi}=\mathrm{col}(\bar{\xi}_{1}$, $\cdots $, $\bar{\xi}_{N})$%
, $\zeta _{1}=\mathrm{col}(\zeta _{11}$, $\cdots $, $\zeta _{1N})$, and $%
\zeta _{2}=\mathrm{col}(\zeta _{21}$, $\cdots $, $\zeta _{2N})$. Then, it
follows from (\ref{12}) that\ $e_{v}=He=H\bar{\xi}$, where $H$ is defined in
Remark \ref{Rem1}. Moreover, the closed-loop system (\ref{18}) can be
written in a compact form as%
\begin{align}
\dot{\bar{z}}& =\bar{A}_{1}(w)\bar{z}+\bar{A}_{2}(w)\bar{\xi}\text{,}  \notag
\\
\dot{\bar{\eta}}& =\left( I_{N}\otimes M\right) \bar{\eta}+\bar{A}_{6}(w)%
\bar{z}+\bar{A}_{7}(w)\bar{\xi}\text{,}  \notag \\
\dot{\bar{\xi}}& =\bar{A}_{3}(w)\bar{z}+\bar{A}_{5}(w)\bar{\xi}-b(w)Ke_{v}
\notag \\
& \quad +(b(w)\otimes \Psi _{\sigma })\bar{\eta}+\bar{A}_{8}(w)\eta
+b(w)\zeta _{2}-b(w)\zeta _{1}\text{,}  \notag \\
\dot{\hat{\Psi}}_{i}^{\mathrm{T}}& =-\gamma _{i}e_{vi}\eta _{i}\text{,}
\notag \\
\dot{K}_{i}& =\delta _{i}e_{vi}^{2}\text{, }  \notag \\
\dot{h}_{i}& =-\alpha _{i}h_{i}-f_{i}(\zeta _{1i},\zeta _{2i},e_{vi},t)\text{%
, }i=1\text{, }\cdots \text{, }N\text{,}  \label{19}
\end{align}%
where $\bar{A}_{s}(w)=\mathrm{blockdiag}\{A_{s1}(w)$, $\cdots $, $%
A_{sN}(w)\} $, $s=1,2,3,5,6,7$, $\bar{A}_{8}(w)=\mathrm{blockdiag}\{b_{1}(w)%
\tilde{\Psi}_{1}$, $\cdots $, $b_{N}(w)\tilde{\Psi}_{N}\}$, $b(w)=\mathrm{%
diag}\{b_{1}(w) $, $\cdots $, $b_{N}(w)\}$, and $K=\mathrm{diag}\{K_{1}$, $%
\cdots $, $K_{N}\} $.

Before presenting the solution to the considered problem, a technical lemma
is first established.

\begin{lemma}
\label{Lem1}Consider the closed-loop system\ (\ref{19}) under the
distributed event-triggering mechanism (\ref{17}). Suppose that Assumptions %
\ref{Asmp1}-\ref{Asmp3} are satisfied. Then, for any $w\in \mathbb{W}$ and $%
\sigma \in \mathbb{S}$, all signals of the closed-loop system (\ref{19}) are
bounded for all $t\geq 0$, and the trigger variables satisfy $h_{i}(t)>0$, $%
i=1$, $\cdots $, $N$. Moreover, the Zeno behavior is ruled out.
\end{lemma}

\begin{proof}
The proof is composed of the following three steps.

\emph{Step 1:} we evaluate the trigger variable $h_{i}$. The triggering
mechanism (\ref{17}) ensures that for $t\in \lbrack t_{k}^{i},t_{k+1}^{i})$,
\begin{equation}
f_{i}(\zeta _{1i},\zeta _{2i},e_{vi},t)\leq h_{i}(t)\text{.}  \label{Eq1}
\end{equation}%
This, together with (\ref{dynamic}), implies that
\begin{equation}
\dot{h}_{i}(t)\geq -\left( \alpha _{i}+1\right) h_{i}(t)\text{, }h_{i}(0)>0%
\text{.}  \label{Eq2}
\end{equation}%
Then, by the comparison lemma (\cite[ Lemma 3.4]{Khalil2002}), one can obtain from (\ref{Eq2}) that
\begin{equation}
h_{i}(t)\geq \text{ }h_{i}(0)e^{-\left( \alpha _{i}+1\right) t}>0\text{, }%
t\in \lbrack t_{k}^{i},t_{k+1}^{i})\text{.}  \label{Eq3}
\end{equation}

\emph{Step 2: }we analyze the stability of system (\ref{19}). Under
Assumption \ref{Asmp2}, $A_{1i}(w)$, $i=1$, $\cdots $, $N$, are Hurwitz for
all $w\in \mathbb{W}$, and so is $\bar{A}_{1}(w)$. Then, for all $w\in
\mathbb{W}$, the Lyapunov equation
\begin{equation}
\bar{A}_{1}^{\mathrm{T}}(w)P_{1}(w)+P_{1}(w)\bar{A}_{1}(w)=-2I  \label{26}
\end{equation}%
admits a positive definite solution $P_{1}(w)$. Since $M$ is Hurwitz, the
Lyapunov equation
\begin{equation}
\left( I_{N}\otimes M\right) ^{\mathrm{T}}P_{2}+P_{2}\left( I_{N}\otimes
M\right) =-2I  \label{27}
\end{equation}%
admits a positive definite solution $P_{2}$.

To analyze the stability of system (\ref{19}), consider the following
Lyapunov-like function:%
\begin{align}
V& =\mu _{0}\bar{z}^{\mathrm{T}}P_{1}(w)\bar{z}+\bar{\eta}^{\mathrm{T}}P_{2}%
\bar{\eta}+\bar{\xi}^{\mathrm{T}}H\bar{\xi}  \notag \\
& \quad +\sum_{i=1}^{N}\frac{1}{\gamma _{i}}b_{i}(w)\tilde{\Psi}_{i}\tilde{%
\Psi}_{i}^{\mathrm{T}}+\sum_{i=1}^{N}\frac{1}{\delta _{i}}b_{i}(w)\tilde{K}%
_{i}^{2}+\sum_{i=1}^{N}h_{i}\text{,}  \label{28}
\end{align}%
with $\tilde{\Psi}_{i}\triangleq \hat{\Psi}_{i}-\Psi _{\sigma }$ and $\tilde{%
K}_{i}\triangleq K_{i}-K_{0}$, where $\mu _{0}$ and $K_{0}$ are positive
constants to be specified later. It should be noted that $V$ is not
differentiable at the triggering times of all agents because $\bar{\xi}$ is
not differentiable at these time instants. The time derivative of $V$ along
the trajectory of system (\ref{19}) is
\begin{align}
\dot{V}=& \ 2\mu _{0}\bar{z}^{\mathrm{T}}P_{1}(w)\left( \bar{A}_{1}(w)\bar{z}%
+\bar{A}_{2}(w)\bar{\xi}\right)  \notag \\
& +2\bar{\eta}^{\mathrm{T}}P_{2}\left[ \left( I_{N}\otimes M\right) \bar{\eta%
}+\bar{A}_{6}(w)\bar{z}+\bar{A}_{7}(w)\bar{\xi}\right]  \notag \\
& +2\bar{\xi}^{\mathrm{T}}H\left[ \bar{A}_{3}(w)\bar{z}+\bar{A}_{5}(w)\bar{%
\xi}+(b(w)\otimes \Psi _{\sigma })\bar{\eta}\right]  \notag \\
& -2\bar{\xi}^{\mathrm{T}}Hb(w)Ke_{v}+2\bar{\xi}^{\mathrm{T}}H\bar{A}%
_{8}(w)\eta +2\bar{\xi}^{\mathrm{T}}Hb(w)\zeta _{2}  \notag \\
& -2\bar{\xi}^{\mathrm{T}}Hb(w)\zeta _{1}-2e_{v}^{\mathrm{T}}\bar{A}%
_{8}(w)\eta  \notag \\
& +2e_{v}^{\mathrm{T}}b(w)(K-K_{0}I_{N})e_{v}  \notag \\
& -\sum_{i=1}^{N}\left( \alpha _{i}h_{i}+\zeta _{1i}^{2}+\zeta
_{2i}^{2}-\kappa _{i}e_{vi}^{2}-\beta _{i}\right) \text{,}  \label{29a}
\end{align}%
where $\dot{V}(t)$ represents the right-hand derivative of $V(t)$ when $%
t=t_{k}^{i}$, for all $k\in \mathbb{N}$ and $i=1$, $\cdots $, $N$. Since $%
e_{v}=H\bar{\xi}$ and $H=H^{\mathrm{T}}$, one has $2e_{v}^{\mathrm{T}%
}b(w)Ke_{v}=2\bar{\xi}^{\mathrm{T}}Hb(w)Ke_{v}$. Then, based on (\ref{26})
and (\ref{27}), $\dot{V}$ in (\ref{29a}) becomes
\begin{align}
\dot{V}=& -2\mu _{0}\left\Vert \bar{z}\right\Vert ^{2}+2\mu _{0}\bar{z}^{%
\mathrm{T}}P_{1}(w)\bar{A}_{2}(w)\bar{\xi}-2\left\Vert \bar{\eta}\right\Vert
^{2}  \notag \\
& +2\bar{\eta}^{\mathrm{T}}P_{2}\bar{A}_{6}(w)\bar{z}+2\bar{\eta}^{\mathrm{T}%
}P_{2}\bar{A}_{7}(w)\bar{\xi} \notag \\
&+2\bar{\xi}^{\mathrm{T}}H\bar{A}_{3}(w)\bar{z}
+2\bar{\xi}^{\mathrm{T}}H\bar{A}_{5}(w)\bar{\xi}  \notag \\
&-2K_{0}e_{v}^{\mathrm{T}
}b(w)e_{v}+2\bar{\xi}^{\mathrm{T}}H(b(w)\otimes \Psi _{\sigma })\bar{\eta}
\notag \\
& +2\bar{\xi}^{\mathrm{T}}Hb(w)\zeta _{2}-2\bar{\xi}^{\mathrm{T}}Hb(w)\zeta
_{1}-\left\Vert \zeta _{1}\right\Vert ^{2}-\left\Vert \zeta _{2}\right\Vert
^{2}  \notag \\
& -\sum_{i=1}^{N}\alpha _{i}h_{i}+\sum_{i=1}^{N}\left( \kappa
_{i}e_{vi}^{2}+\beta _{i}\right) \text{.}  \label{29b}
\end{align}%
Using the Young's inequality, one can obtain the following inequalities:%
\begin{align*}
& 2\mu _{0}\bar{z}^{\mathrm{T}}P_{1}(w)\bar{A}_{2}(w)\bar{\xi}\leq \mu
_{0}\left\Vert \bar{z}\right\Vert ^{2} \notag \\
&\qquad\qquad\qquad\qquad  +\mu _{0}\left\Vert P_{1}(w)\bar{A}%
_{2}(w)H^{-1}\right\Vert ^{2}\left\Vert e_{v}\right\Vert ^{2}\text{,} \\
& 2\bar{\eta}^{\mathrm{T}}P_{2}\bar{A}_{6}(w)\bar{z}\leq 2\left\Vert P_{2}%
\bar{A}_{6}(w)\right\Vert ^{2}\left\Vert \bar{z}\right\Vert ^{2}+\frac{1}{2}%
\left\Vert \bar{\eta}\right\Vert ^{2}\text{,} \\
& 2\bar{\eta}^{\mathrm{T}}P_{2}\bar{A}_{7}(w)\bar{\xi}\leq 4\left\Vert P_{2}%
\bar{A}_{7}(w)H^{-1}\right\Vert ^{2}\left\Vert e_{v}\right\Vert ^{2}+\frac{1%
}{4}\left\Vert \bar{\eta}\right\Vert ^{2}\text{,} \\
& 2\bar{\xi}^{\mathrm{T}}H\bar{A}_{3}(w)\bar{z}\leq \left\Vert \bar{A}%
_{3}(w)\right\Vert ^{2}\left\Vert e_{v}\right\Vert ^{2}+\left\Vert \bar{z}%
\right\Vert ^{2}\text{,} \\
& 2\bar{\xi}^{\mathrm{T}}H(b(w)\otimes \Psi _{\sigma })\bar{\eta}\leq
4\left\Vert b(w)\otimes \Psi _{\sigma }\right\Vert ^{2}\left\Vert
e_{v}\right\Vert ^{2}+\frac{1}{4}\left\Vert \bar{\eta}\right\Vert ^{2}\text{,%
} \\
& 2\bar{\xi}^{\mathrm{T}}Hb(w)\zeta _{2}\leq \left\Vert b(w)\right\Vert
^{2}\left\Vert e_{v}\right\Vert ^{2}+\left\Vert \zeta _{2}\right\Vert ^{2}%
\text{,} \\
& -2\bar{\xi}^{\mathrm{T}}Hb(w)\zeta _{1}\leq \left\Vert b(w)\right\Vert
^{2}\left\Vert e_{v}\right\Vert ^{2}+\left\Vert \zeta _{1}\right\Vert ^{2}%
\text{.}
\end{align*}%
Substituting these inequalities into (\ref{29b}) yields%
\begin{align}
\dot{V}& \leq -\big(\mu _{0}-2\left\Vert P_{2}\bar{A}_{6}(w)\right\Vert
^{2}-1\big)\left\Vert \bar{z}\right\Vert ^{2}-\left\Vert \bar{\eta}%
\right\Vert ^{2}  \notag \\
& \quad +\left( \mu _{1}(w)-2K_{0}b_{\min }(w)\right) \left\Vert
e_{v}\right\Vert ^{2}  \notag \\
&\quad -\sum_{i=1}^{N}\alpha _{i}h_{i}+\sum_{i=1}^{N}\left( \kappa
_{i}e_{vi}^{2}+\beta _{i}\right) \text{,}  \label{30}
\end{align}%
where $\mu _{1}(w)\triangleq \left\Vert \bar{A}_{5}(w)H^{-1}\right\Vert +\mu
_{0}\left\Vert P_{1}(w)\bar{A}_{2}(w)H^{-1}\right\Vert ^{2}+4\left\Vert P_{2}%
\bar{A}_{7}(w)H^{-1}\right\Vert ^{2}+\left\Vert \bar{A}_{3}(w)\right\Vert
^{2}+4\left\Vert b(w)\otimes \Psi _{\sigma }\right\Vert ^{2}+2\left\Vert
b(w)\right\Vert ^{2}$, and $b_{\min }(w)\triangleq \min_{i=1,\cdots
,N}b_{i}(w)>0$. Let $\mu _{0}$ and $K_{0}$ be chosen such that $\mu _{0}\geq
\sup_{w\in \mathbb{W}}$ $\big(2\left\Vert P_{2}\bar{A}_{6}(w)\right\Vert
^{2}+2\big)$ and $K_{0}\geq \sup_{w\in \mathbb{W}}[(\mu _{1}(w)+1)/(2b_{\min
}(w))]$, respectively. Then, it follows from (\ref{30}) that%
\begin{equation}
\dot{V}\leq -\left\Vert \bar{z}\right\Vert ^{2}-\left\Vert \bar{\eta}%
\right\Vert ^{2}-\left\Vert e_{v}\right\Vert ^{2}-\sum_{i=1}^{N}\alpha
_{i}h_{i}+\sum_{i=1}^{N}\left( \kappa _{i}e_{vi}^{2}+\beta _{i}\right) \text{%
.}  \label{31}
\end{equation}%
Let
\begin{align*}
&\kappa _{\max }\triangleq \max_{i=1,\cdots ,N}\kappa _{i}<1\text{, }\alpha
_{\min }\triangleq \min_{i=1,\cdots ,N}\alpha _{i}\text{, } \notag \\
&\beta _{\max}\triangleq \max_{i=1,\cdots ,N}\beta _{i}\text{.}
\end{align*}%
Then, one can obtain from (\ref{31}) that%
\begin{align}
\dot{V}& \leq -\left\Vert \bar{z}\right\Vert ^{2}-\left\Vert \bar{\eta}%
\right\Vert ^{2}-\left( 1-\kappa _{\max }\right) \left\Vert e_{v}\right\Vert
^{2}  \notag \\
&\quad -\alpha _{\min }\sum_{i=1}^{N}h_{i}+N\beta _{\max }  \notag \\
& \leq -(1-\kappa _{\max })\big(\left\Vert \bar{z}\right\Vert
^{2}+\left\Vert \bar{\eta}\right\Vert ^{2}+\left\Vert e_{v}\right\Vert ^{2}%
\big)+N\beta _{\max }\text{.}  \label{32}
\end{align}%
This implies that $\dot{V}<0$ as long as the following condition is
satisfied:
\begin{equation}
\left\Vert \bar{z}\right\Vert ^{2}+\left\Vert \bar{\eta}\right\Vert
^{2}+\left\Vert e_{v}\right\Vert ^{2}>\frac{N\beta _{\max }}{(1-\kappa
_{\max })}\text{.}  \label{33}
\end{equation}%
Thus, from the definition of $V$, it is obtained that the signals $\bar{z}$,
$\bar{\eta}$, $\bar{\xi}$, $\tilde{\Psi}_{i}$,\ $\tilde{K}_{i}$, $h_{i}$, $%
i=1$, $\cdots $, $N$, are all bounded on the time interval $[0$, $\infty )$.

\emph{Step 3:} we show the exclusion of the Zeno behavior. To this end, we
consider the following static triggering mechanism corresponding to the
dynamic one in (\ref{17}):
\begin{equation}
t_{k+1}^{i}=\inf \left\{ t>t_{k}^{i}\mid f_{i}(\zeta _{1i},\zeta
_{2i},e_{vi},t)\geq 0\right\} \text{.}  \label{34}
\end{equation}%
Clearly, this static triggering mechanism yields a more conservative
inter-event time compared to the dynamic one. Thus, in order to show the
exclusion of the Zeno behavior under the dynamic triggering mechanism (\ref%
{17}), it suffices to show that the corresponding static triggering
mechanism (\ref{33}) can guarantee a positive lower bound of inter-event
times.

Let $\zeta _{i}=\mathrm{col}(\zeta _{1i}$, $\zeta _{2i})$. Then, based on
the definitions of $\zeta _{1i}$ and $\zeta _{2i}$, the dynamics of $%
\left\Vert \zeta _{i}\right\Vert $ on every interval $[t_{k}^{i}$, $%
t_{k+1}^{i})$, $k\in \mathbb{N}$ satisfy%
\begin{align}
\frac{d}{dt}\left\Vert \zeta _{i}\right\Vert & =\frac{\zeta _{i}^{\mathrm{T}}%
\dot{\zeta}_{i}}{\left\Vert \zeta _{i}\right\Vert }\leq \left\Vert \dot{\zeta%
}_{i}\right\Vert \leq \left\vert \dot{\zeta}_{1i}\right\vert +\left\vert
\dot{\zeta}_{2i}\right\vert  \notag \\
& =\big\vert\dot{K}_{i}e_{vi}+K_{i}\dot{e}_{vi}\big\vert+\big\vert\dot{\hat{%
\Psi}}_{i}\eta _{i}+\hat{\Psi}_{i}\dot{\eta}_{i}\big\vert  \notag \\
& \leq \delta _{i}\left\vert e_{vi}\right\vert ^{3}+\left\vert
K_{i}\right\vert \left\vert \dot{e}_{vi}\right\vert +\gamma _{i}\left\vert
e_{vi}\right\vert \left\Vert \eta _{i}\right\Vert ^{2}+\big\vert\hat{\Psi}%
_{i}\dot{\eta}_{i}\big\vert\text{.}  \label{20}
\end{align}%
From the coordinate transformation of $\bar{\eta}_{i}$ in (\ref{10}), it
follows that%
\begin{equation}
\eta _{i}=\bar{\eta}_{i}+\bar{\Upsilon}_{i}(w,\sigma )v+\frac{1}{b_{i}(w)}Q%
\bar{\xi}_{i}\text{.}  \label{21}
\end{equation}%
Since $e_{v}=\mathrm{col}(e_{v1},\cdots ,e_{vN})=H\bar{\xi}$, then for $i=1$%
, $\cdots $, $N$,
\begin{align}
\left\vert e_{vi}\right\vert & \leq \left\Vert e_{v}\right\Vert \leq
\left\Vert H\right\Vert \left\Vert \bar{\xi}\right\Vert \text{,}  \label{22}
\\
\left\vert \dot{e}_{vi}\right\vert & \leq \left\Vert \dot{e}_{v}\right\Vert
\leq \left\Vert H\right\Vert \left\Vert \dot{\bar{\xi}}\right\Vert \text{.}
\label{23}
\end{align}%
Substituting (\ref{21})-(\ref{23}) into (\ref{20}) yields%
\begin{align}
\frac{d}{dt}\left\Vert \zeta _{i}\right\Vert & \leq \delta _{i}\left\Vert
H\right\Vert ^{3}\left\Vert \bar{\xi}\right\Vert ^{3}+\left\vert
K_{i}\right\vert \left\Vert H\right\Vert \left\Vert \dot{\bar{\xi}}%
\right\Vert  \notag \\
& \quad +\gamma _{i}\left\Vert H\right\Vert \left\Vert \bar{\xi}\right\Vert
\left\Vert \eta _{i}\right\Vert ^{2}  \notag \\
& \quad +\left\vert \hat{\Psi}_{i}\left( \dot{\bar{\eta}}_{i}+\bar{\Upsilon}%
_{i}(w,\sigma )Sv+\frac{1}{b_{i}(w)}Q\dot{\bar{\xi}}_{i}\right) \right\vert
\text{,}  \label{24}
\end{align}%
for $t\in \lbrack t_{k}^{i}$, $t_{k+1}^{i})$. Recall that the signals $\bar{z%
}$, $\bar{\eta}$, $\bar{\xi}$, $\tilde{\Psi}_{i}$,\ $\tilde{K}_{i}$, $h_{i}$%
, $i=1$, $\cdots $, $N$, are all bounded. This indicates that $e_{vi}$, $i=1$%
, $\cdots $, $N$, are also bounded, as $e_{v}=H\bar{\xi}$. As mentioned in
Remark \ref{Rem1}, the signal $v$ is bounded. Since $w$ and $\sigma $ belong
to the compact sets, $\bar{\Upsilon}_{i}(w,\sigma )$ and $b_{i}(w)$ are
bounded. Hence, it follows from (\ref{21}) that $\eta _{i}$ is bounded.
Then, the boundedness of $K_{i}$, $e_{vi}$, $\hat{\Psi}_{i}$, and $\eta _{i}$
implies that $\zeta _{1i}$ and $\zeta _{2i}$ are bounded, for $i=1$, $\cdots
$, $N$. Based on the previous discussion, it can be seen from (\ref{19})
that $\dot{\bar{z}}$, $\dot{\bar{\xi}}$, and $\dot{\bar{\eta}}$ are bounded.
Consequently, it is concluded from (\ref{24}) that $(d/dt)\left\Vert \zeta
_{i}\right\Vert $ is bounded, which means that there exists a positive
constant $\varpi _{i}$ such that $(d/dt)\left\Vert \zeta _{i}\right\Vert
\leq $ $\varpi _{i}$. To ensure that the condition $f_{i}(\zeta
_{1i},\zeta _{2i},e_{vi},t)\leq 0$ holds, a sufficient condition is given by
$\zeta _{1i}^{2}(t)+\zeta _{2i}^{2}(t)\leq \beta _{i}$, i.e., $\left\Vert
\zeta _{i}\right\Vert ^{2}\leq \beta _{i}$. It follows that on the interval $%
[t_{k}^{i}$, $t_{k+1}^{i})$,
\begin{equation}
\varpi _{i}\left( t_{k+1}^{i}-t_{k}^{i}\right) \geq \sqrt{\beta _{i}}\text{.}
\label{25a}
\end{equation}%
which implies that
\begin{equation}
t_{k+1}^{i}-t_{k}^{i}\geq \frac{\sqrt{\beta _{i}}}{\varpi _{i}}>0\text{.}\label{25b}
\end{equation}%
Consequently, the exclusion of the Zeno behavior is guaranteed under the static
triggering mechanism (\ref{33}), and hence, it extends to the dynamic triggering mechanism (\ref{17}).
\end{proof}


\begin{remark}
\label{Rem4}Using $e_{v}=H\bar{\xi}$, one can obtain from (\ref{33}) that
the ultimate bound of $\bar{\xi}$ is given by $\frac{N\beta _{\max }}{%
\lambda _{\min }(H^{2})(1-\kappa _{\max })}$, where $\lambda _{\min }(H^{2})$
denotes the minimum eigenvalue of $H^{2}$. Hence, the ultimate bound of the
tracking error $e_{i}$ is given by $\frac{N\beta _{\max }}{\lambda _{\min
}(H^{2})(1-\kappa _{\max })}$, as $e_{i}=\bar{\xi}_{i}$, $i=1$, $\cdots $, $%
N $. As seen, the ultimate bound of the tracking error $e_{i}$ can be made
arbitrarily small by choosing the value of $\beta _{i}$ sufficiently small. Thus, the level of accuracy $\varepsilon _{i}$ given in Definition \ref{CRORP} can be arbitrarily specified.
\end{remark}

Now, we are ready to present the solution to the robust CORP under consideration.

\begin{theorem}
\label{Thm1}Under Assumptions \ref{Asmp1}-\ref{Asmp3}, the robust CORP of the MAS (\ref{1}) and exosystem (\ref{2}%
) is solvable by the following distributed event-triggered adaptive output
feedback control law:
\begin{align}
u_{i}(t) &=\hat{\Psi}_{i}(t_{k}^{i})\eta
_{i}(t_{k}^{i})-K_{i}(t_{k}^{i})e_{vi}(t_{k}^{i})\text{,}  \notag \\
\dot{\eta}_{i}(t) &=M\eta _{i}(t)+Q\left( \hat{\Psi}_{i}(t_{k}^{i})\eta
_{i}(t_{k}^{i})-K_{i}(t_{k}^{i})e_{vi}(t_{k}^{i})\right) \text{, }  \notag \\
\dot{\hat{\Psi}}_{i}^{\mathrm{T}}(t) &=-\gamma _{i}e_{vi}(t)\eta _{i}(t)%
\text{,}  \notag \\
\dot{K}_{i}(t) &=\delta _{i}e_{vi}^{2}(t)\text{, }  \label{37}
\end{align}%
under the distributed dynamic event-triggering mechanism (\ref{17}). Moreover, the Zeno
behavior is ruled out for each agent.
\end{theorem}

\begin{proof}
Note that under the input transformation in (\ref{10}), the control law (\ref%
{37}) is equivalent to the control law (\ref{13}). By Lemma \ref{Lem1},
under the control law (\ref{13}) and the triggering mechanism (\ref{17}),
the closed-loop signals $\bar{z}_{i}$, $\bar{\eta}_{i}$, $\bar{\xi}_{i}$, $%
\hat{\Psi}_{i}$,\ $K_{i}$, $i=1$, $\cdots $, $N$, are bounded on $[0$, $%
\infty )$. Since the signal $v$ is bounded on $[0$, $\infty )$ and the
parameters $w$ and $\sigma $ belong to any given compact sets, then from the
coordinate transformation in (\ref{10}), it follows that the signals $z_{i}$%
, $\eta _{i}$, $\xi _{i}$, $i=1$, $\cdots $, $N$, are bounded on $[0$, $%
\infty )$. Moreover, according to Lemma \ref{Lem1} and Remark \ref{Rem4}, $%
\lim_{t\rightarrow \infty }\left\vert e_{i}(t)\right\vert \leq \varepsilon_{i}\triangleq \frac{N\beta
_{\max }}{\lambda _{\min }(H^{2})(1-\kappa _{\max })}$, for $i=1$, $\cdots $%
, $N$. Therefore, the considered problem can be solved under the control law
(\ref{37}) and the triggering mechanism (\ref{17}). In addition, by Lemma %
\ref{Lem1}, the Zeno behavior is ruled out for each agent. Thus, the proof
is completed.
\end{proof}

\begin{remark}
\label{Rem5}The distributed internal model in the second equation of (\ref%
{37}) contains the parameter estimate $\hat{\Psi}_{i}$ and the adaptive gain
$K_{i}$, and thus is called a distributed adaptive internal model. Based on
this distributed adaptive internal model, the distributed event-triggered
control input $u_{i}(t)$ is constructed. It should be pointed out that $%
u_{i}(t)$ is piecewise constant, which makes it easy to be implemented in a
zero-order-hold manner. Moreover, under the triggering mechanism (\ref{17}),
agent $i$ updates its control input $u_{i}(t)$ only at its own triggering
time instants, regardless of the triggering time instants of its neighboring
agents. As a consequence, the number of controller updates is reduced for
each agent.
\end{remark}

\begin{remark}
\label{Rem6}It is observed that the control law (\ref{13}) and the
triggering mechanism (\ref{17}) depend on the continuous measurement of $%
e_{vi}$. This requires continuous sensing of relative output or continuous communication of absolute output between neighboring agents. To reduce the
communication load when $e_{vi}$ is accessed via inter-agent communication, a sampling mechanism or an estimation scheme for $e_{vi}$
needs to be developed. Specifically, one possible way is to design a
self-triggered sampling mechanism. In the self-triggered control, the
control law and the triggering mechanism are designed using the sampled
information $e_{vi}(t_{k}^{i})$, instead of the real-time information $%
e_{vi}(t)$. In this case, continuous measurement of neighboring information
can be avoided. Another possible way is to design an estimation scheme in
which each agent is required to continuously estimate its neighboring
information so as to relax the continuous measurement requirement.
\end{remark}

\section{A Simulation Example}

Consider a heterogeneous uncertain linear MAS of the form
\begin{align}
\dot{x}_{i}& =A_{i}(w)x_{i}+B_{i}(w)u_{i}+E_{i}(w)v\text{,}  \notag \\
y_{i}& =C_{i}(w)x_{i}\text{, } i=1\text{, }\cdots \text{, }4\text{,}
\label{38}
\end{align}%
where the system matrices are given as follows: for $i=1$, $2$,
\begin{align*}
A_{i}(w)& =\left[
\begin{array}{ccc}
c_{1i}(w) & 1 & 0 \\
0 & -1 & 1 \\
1 & c_{2i}(w) & c_{3i}(w)%
\end{array}%
\right] \text{,} \\
B_{i}(w)& =\left[
\begin{array}{ccc}
0 & 0 & c_{4i}(w)%
\end{array}%
\right] ^{\mathrm{T}}\text{, } \\
C_{i}(w)& =\left[
\begin{array}{ccc}
0 & 0 & 1%
\end{array}%
\right] \text{, }E_{i}=0_{3\times 2}\text{, }
\end{align*}%
and for $i=3$, $4$,%
\begin{align*}
A_{i}(w)& =\left[
\begin{array}{cccc}
c_{1i}(w) & 0 & 0 & 1 \\
0 & -1 & 0 & 1 \\
0 & 0 & -2 & 1 \\
1 & 1 & c_{2i}(w) & c_{3i}(w)%
\end{array}%
\right] \text{,} \\
B_{i}(w)& =\left[
\begin{array}{cccc}
0 & 0 & 0 & c_{4i}(w)%
\end{array}%
\right] ^{\mathrm{T}}\text{,} \\
C_{i}(w)& =\left[
\begin{array}{cccc}
0 & 0 & 0 & 1%
\end{array}%
\right] \text{, }E_{i}=0_{4\times 2}\text{.}
\end{align*}%
The parameter vector $c_{i}(w)\triangleq \text{col}(c_{1i}(w)$, $c_{2i}(w)$,
$c_{3i}(w)$, $c_{4i}(w))$ can be expressed as $c_{i}(w)=$ $c_{i}(0)+w_{i}$, $%
i=1$, $\cdots$, $4$, with $c_{i}(0)=[c_{1i}(0)$, $c_{2i}(0)$, $c_{3i}(0)$, $%
c_{4i}(0)]^{\mathrm{T}}$ and $w_{i}=[w_{1i}$, $w_{2i}$, $w_{3i}$, $w_{4i}]^{%
\mathrm{T}}\in \mathbb{W\subset R}^{4}$, where $c_{i}(0)$ and $w_{i}$
represent the nominal part and the uncertainty of $c_{i}(w)$, respectively.
The exogeneous signal $v$ is generated by an unknown exosystem of the form
\begin{equation}
\dot{v}=\left[
\begin{array}{cc}
0 & \sigma \\
-\sigma & 0%
\end{array}%
\right] v\text{, }y_{0}=\left[
\begin{array}{cc}
1 & 0%
\end{array}%
\right] v\text{,}  \label{39}
\end{equation}%
with the unknown scalar parameter $\sigma \in \mathbb{S\subset R}$. The
compact subsets $\mathbb{W}$ and $\mathbb{S}$ are allowed to be unknown. The
communication graph among the agents and the exosystem is displayed in Fig. %
\ref{Fig0}. The aim here is to make the tracking errors converge to the origin with any specified level of accuracy, namely, $\lim_{t\rightarrow \infty}| e_{i}(t)| =\lim_{t\rightarrow \infty}|y_{i}(t)-y_{0}(t)|\leq \varepsilon_{i}$, $i=1$, $\cdots $, $4$.

\begin{figure}[!htb]
\centering
\includegraphics[width=1.5in]{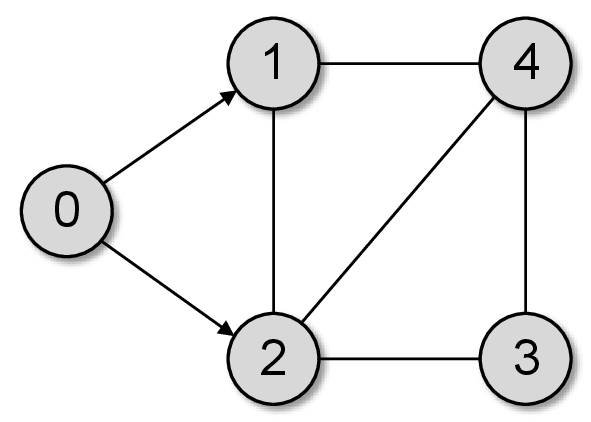}
\caption{Communication graph.}
\label{Fig0}
\end{figure}

By performing $x_{i}=\mathrm{col}(z_{i},\xi _{i})\in \mathbb{R}%
^{n_{i}-1}\times \mathbb{R}$, the MAS (\ref{38}) can be rewritten in the
form of (\ref{1}). Let $c_{1i}(0)=c_{2i}(0)=c_{3i}(0)=-2$, $c_{4i}(0)=2$,
and $c_{1}(w)<0$, for $i=1$, $\cdots $, $4$. It can be verified that
Assumptions \ref{Asmp1}-\ref{Asmp3} hold. Then, by Theorem \ref{Thm1}, the
proposed control strategy is able to solve the robust CORP of the MAS (\ref{38}) and exosystem (\ref{39}). Since $%
S(\sigma )=[0$, $\sigma $; $-\sigma $, $0]$, the minimal polynomial of $%
S(\sigma )$ is $\lambda ^{2}+\sigma ^{2}$. Denote
\begin{equation*}
\Phi (\sigma )=\left[
\begin{array}{cc}
0 & 1 \\
-\sigma ^{2} & 0%
\end{array}%
\right] \text{, }\Gamma =\left[
\begin{array}{cc}
1 & 0%
\end{array}%
\right] \text{.}
\end{equation*}%
Choose the controllable pair $(M$, $Q)$ as%
\begin{equation*}
M=\left[
\begin{array}{cc}
0 & 1 \\
-25 & -10%
\end{array}%
\right] \text{, }Q=\left[
\begin{array}{c}
0 \\
1%
\end{array}%
\right] \text{.}
\end{equation*}%
Solving the Sylvester equation (\ref{7}) yields
\begin{equation*}
T^{-1}(\sigma )=\left[
\begin{array}{cc}
25-\sigma ^{2} & 10 \\
-10\sigma ^{2} & 25-\sigma ^{2}%
\end{array}%
\right] \text{.}
\end{equation*}%
Hence, $\Psi _{\sigma }=\Gamma T^{-1}(\sigma )=[25-\sigma ^{2}$, $10]$.
Denote $\Psi _{\sigma }=[\Psi _{\sigma 1}$, $\Psi _{\sigma 2}]$. Since $%
\sigma $ is unknown, $\Psi _{\sigma 1}=25-\sigma ^{2}$ is also unknown. To
deal with the unknown $\Psi _{\sigma 1}$, an estimator needs to be
constructed for each agent. For agent $i$, $i=1$, $\cdots $, $4$, denote the
estimate of $\Psi _{\sigma 1}$ by $\hat{\Psi}_{1i}$ and the state of the
internal model by $\eta _{i}=[\eta _{1i}$, $\eta _{2i}]^{\mathrm{T}}$. Then,
the proposed distributed ETC law is given by%
\begin{align}
u_{i}(t)& =\hat{\Psi}_{1i}(t_{k}^{i})\eta _{1i}(t_{k}^{i})+10\eta
_{2i}(t_{k}^{i})-K_{i}(t_{k}^{i})e_{vi}(t_{k}^{i})\text{,}  \notag \\
\dot{\eta}_{1i}(t)& =\eta _{2i}(t)\text{,}  \notag \\
\dot{\eta}_{2i}(t)& =-25\eta _{1i}(t)-10\eta _{2i}(t)+\hat{\Psi}%
_{1i}(t_{k}^{i})\eta _{1i}(t_{k}^{i})  \notag \\
& \quad +10\eta _{2i}(t_{k}^{i})-K_{i}(t_{k}^{i})e_{vi}(t_{k}^{i})\text{,}
\notag \\
\dot{\hat{\Psi}}_{1i}(t)& =-\gamma _{i}e_{vi}(t)\eta _{1i}(t)\text{,}  \notag
\\
\dot{K}_{i}(t)& =\delta _{i}e_{vi}^{2}(t)\text{, }  \label{40}
\end{align}%
where $t_{k}^{i}$, $k=0$, $1$, $\cdots $, are determined by the triggering
mechanism (\ref{17}). Without loss of generality, take $t_{0}^{i}=0$, $i=1$,
$\cdots $, $4$.

In the simulation, the actual values of the parameters $\sigma $ and $w_{i}$%
, $i=1$, $\cdots $, $4$, are assumed to be $\sigma =2$, $w_{1}=[0.5$, $1.0$,
$-1.0$, $0.1]^{\mathrm{T}}$, $w_{2}=[-0.5$, $0.5$, $-1.5$, $1.5]^{\mathrm{T}%
} $, $w_{3}=[0.2$, $0.5$, $-0.5$, $1.0]^{\mathrm{T}}$, and $w_{4}=[0.1$, $1.0
$, $-1.0$, $1.5]^{\mathrm{T}}$, and the design parameters are chosen as $%
\gamma _{1}=80$, $\gamma _{2}=\gamma _{3}=\gamma _{4}=10$, $\delta _{i}=5$, $%
\kappa _{i}=0.9$, $\beta _{i}=0.6$, and $\alpha _{i}=1$, $i=1$, $\cdots $, $%
4 $. The simulation results are presented in Figs. \ref{Fig1}-\ref{Fig6},
where the initial condition is chosen as $x_{1}(0)=[-2$, $1$, $-1]^{\mathrm{T%
}}$, $x_{2}(0)=[1$, $-1$, $-2]^{\mathrm{T}}$, $x_{3}(0)=[0$, $2$, $-1$, $2]^{%
\mathrm{T}}$, $x_{4}(0)=[-2$, $2$, $0$, $1]^{\mathrm{T}}$, $v(0)=[0.2$, $1]^{%
\mathrm{T}}$, $\eta _{1}(0)=[-1$, $-2]^{\mathrm{T}}$, $\eta _{2}(0)=[3$, $%
2]^{\mathrm{T}}$, $\eta _{3}(0)=[4$, $6]^{\mathrm{T}}$, $\eta _{4}(0)=[-2$, $%
-4]^{\mathrm{T}}$, $\hat{\Psi}_{1i}(0)=15$, and $K_{i}(0)=10$, $i=1$, $%
\cdots $, $4$. Fig. \ref{Fig1} displays the outputs of the agents and the
exosystem, and Fig. \ref{Fig2} displays the tracking errors of all agents.
From these two figures, it is seen that the output of each agent can
track the output of the exosystem with a high level of accuracy. Shown in Fig. \ref{Fig3}
is the control input of agent $1$. It is observed that the control input of
agent $1$ is piecewise constant, which illustrates the viewpoint in Remark %
\ref{Rem5}. The triggering time instants for updating the control input of
agent $1$ are shown in Fig. \ref{Fig4}. The inter-event times of agents $1$%
, $2$, $3$, $4$ over the time interval $[0,4]$ \textrm{s} are shown in
Fig. \ref{Fig5}. Moreover, the numbers of controller updates in the initial $%
4$ \textrm{s} for agents $1$, $2$, $3$, $4$ are $100$, $445$, $185$, and $%
275$, respectively.

\begin{figure}[!htb]
\centering
\includegraphics[width=3.5in]{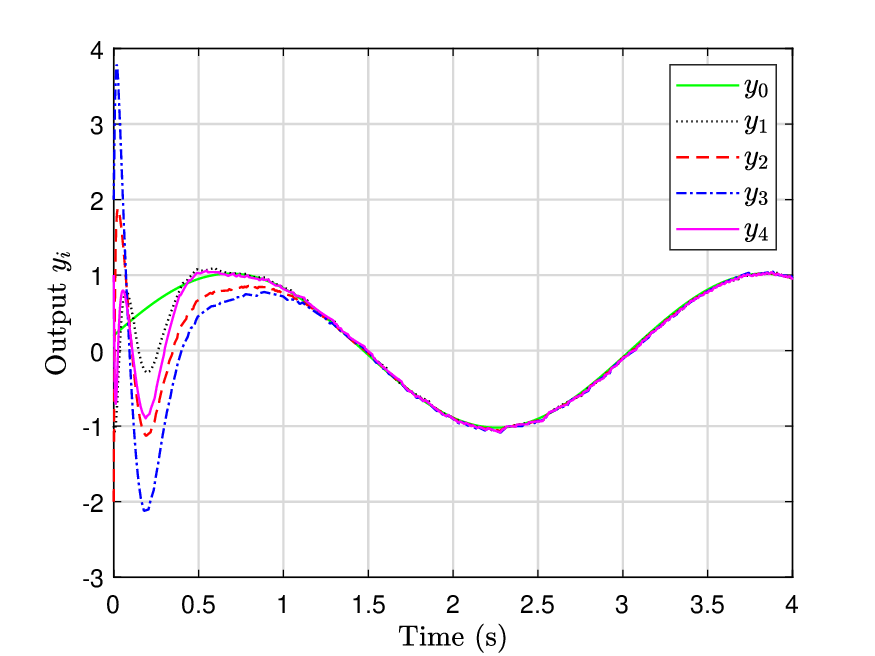}
\caption{Outputs of the agents and the exosystem.}
\label{Fig1}\vspace{-0.5em}
\end{figure}

\begin{figure}[!htb]
\centering
\includegraphics[width=3.5in]{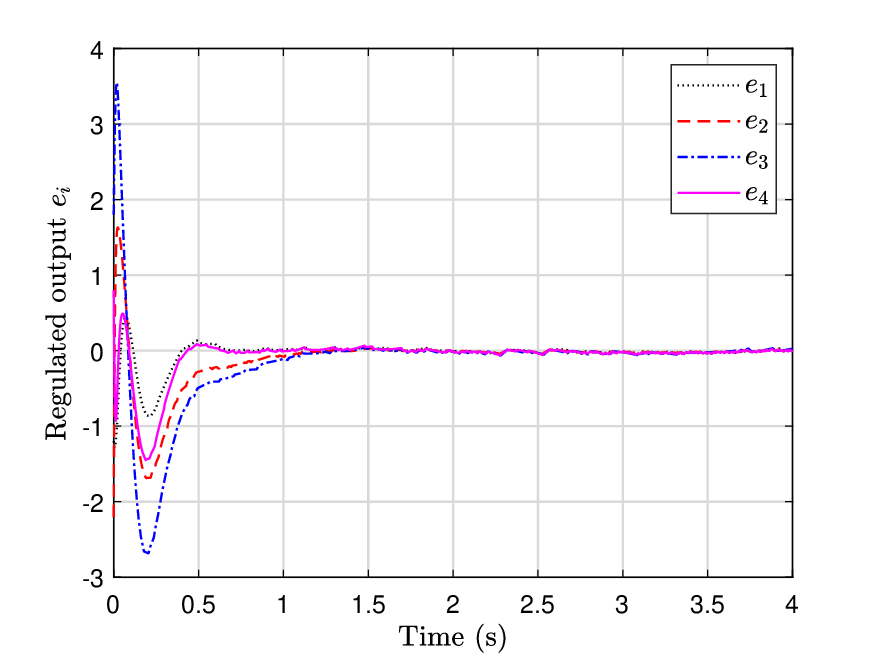}
\caption{Tracking errors of the agents.}
\label{Fig2}
\end{figure}

\begin{figure}[!htb]
\centering
\includegraphics[width=3.5in]{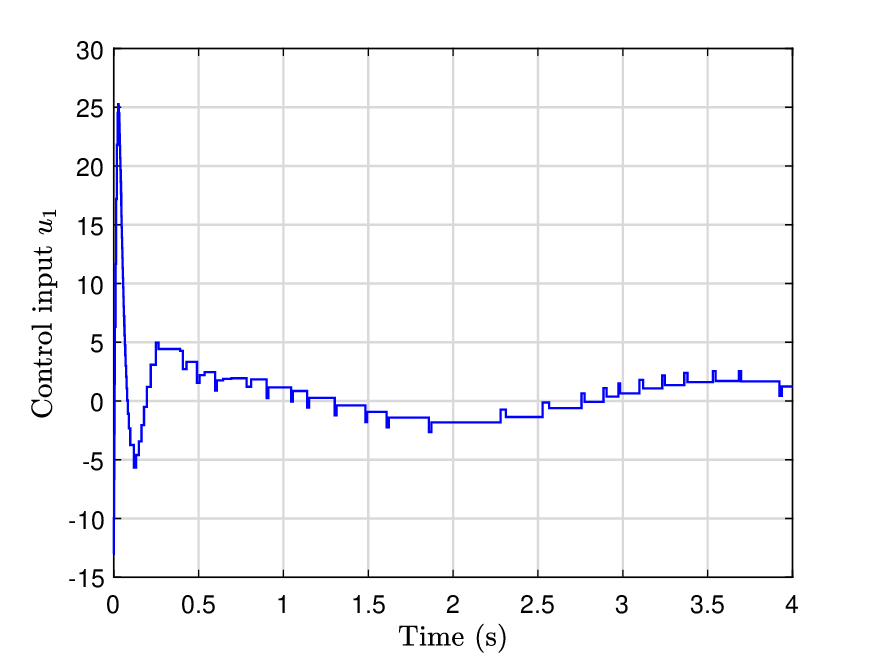}
\caption{Control input of agent $1$.}
\label{Fig3}
\end{figure}

\begin{figure}[!htb]
\centering
\includegraphics[width=3.5in]{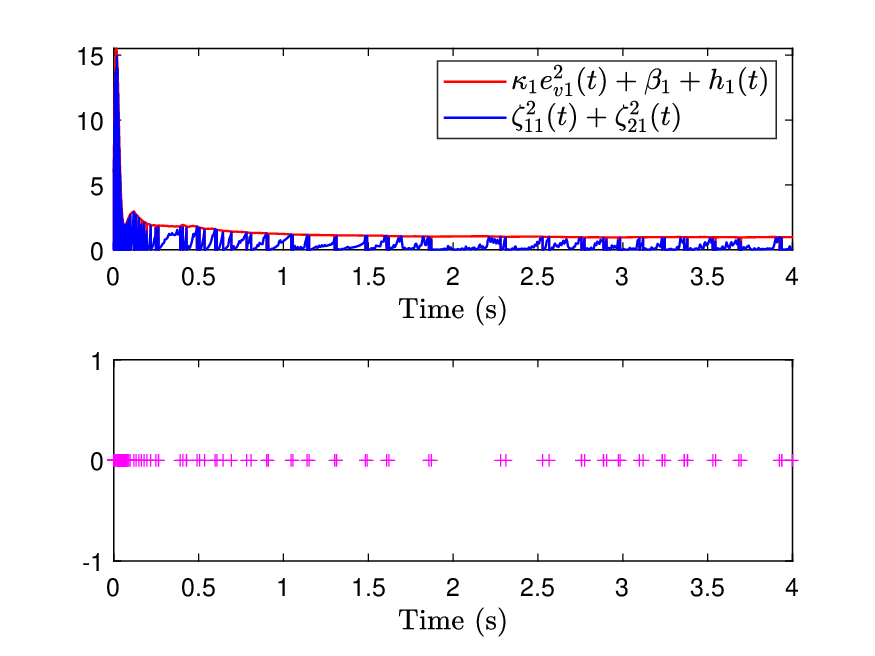}
\caption{Total measurement error, triggering threshold, and triggering time
instants of agent $1$.}
\label{Fig4}
\end{figure}

\begin{figure}[!htb]
\centering
\includegraphics[width=3.5in]{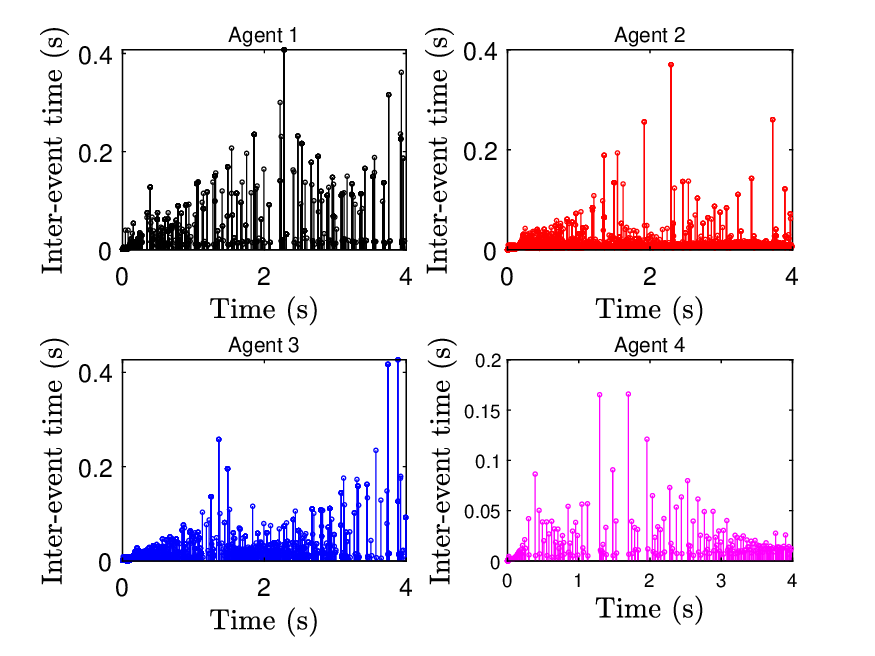}
\caption{Inter-event times of the agents.}
\label{Fig5}
\end{figure}

To illustrate the results in Lemma \ref{Lem1}, the adaptive gains of all
agents are shown in Fig. \ref{Fig6}, and the trigger variables of all agents are presented in Fig. %
\ref{Fig7}. It can be seen that the adaptive
gains are continuous and bounded and the trigger variables are positive, as shown in Lemma \ref{Lem1}.
\begin{figure}[!htb]
\centering
\includegraphics[width=3.5in]{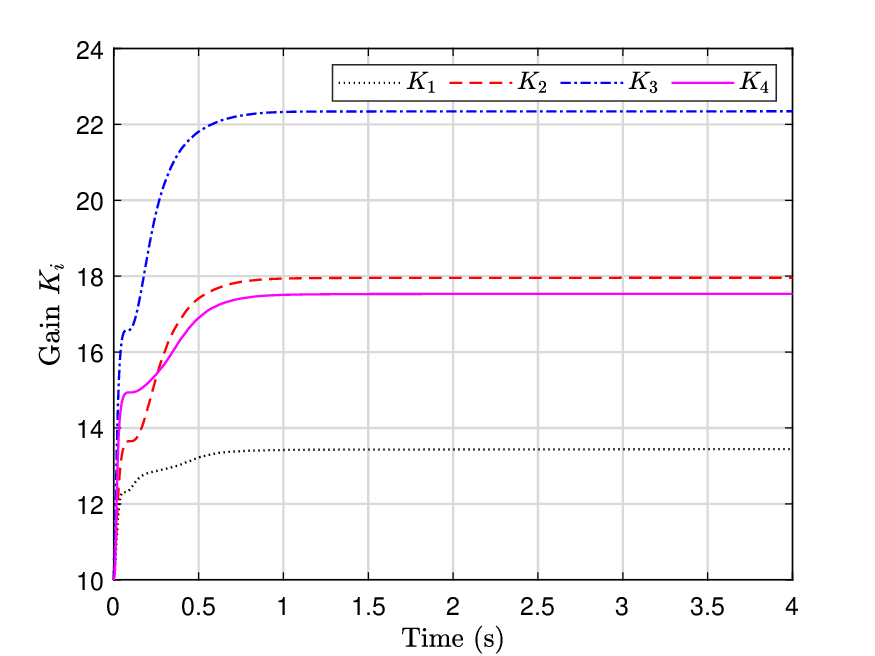}
\caption{Adaptive gains of the agents.}
\label{Fig6}
\end{figure}

\begin{figure}[!htb]
\centering
\includegraphics[width=3.5in]{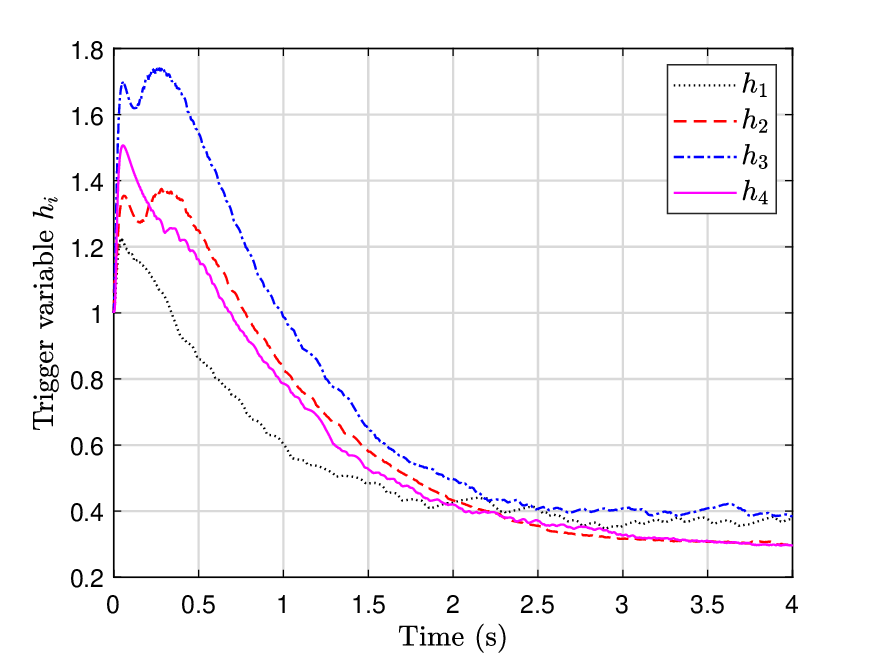}
\caption{Trigger variables of the agents.}
\label{Fig7}
\end{figure}

Next, a performance comparison is made between the proposed ETC strategy and
the typical periodic sampling control strategy. For this purpose, these two
control strategies are carried out under the same average frequency of
controller updates. It is obtained in the above simulation that under the
proposed ETC strategy, the average inter-event times for agents $1$, $2$, $3$%
, $4$ are $0.0398$, $0.0090$, $0.0090$, and $0.0090$ \textrm{s}, respectively. When the same periodic
sampling periods are respectively chosen for agents $1$, $2$, $3$, $4$, the
periodic sampling control strategy would render the closed-loop system
unstable. This illustrates that the proposed ETC strategy has a better
performance in reducing the number of controller updates than the
corresponding periodic sampling control strategy.

Finally, a comparative analysis is conducted between the proposed dynamic
triggering mechanism and the corresponding static triggering mechanism of
the form (\ref{34}). Under the static triggering mechanism (\ref{34}), the
numbers of controller updates within the initial $4$ \textrm{s} for agents $1$, $%
2$, $3$, $4$ are $127 $, $419$, $251$, and $291$, respectively, resulting in a total count of $1088$. In contrast, the total count of controller updates under the proposed dynamic triggering mechanism amounts to $1005$. This comparative analysis substantiates the efficacy of the proposed dynamic triggering mechanism in reducing the frequency of controller updates compared to its static counterpart.

\section{Conclusions\label{Sec5}}

In this paper, the robust CORP has been addressed for a class of
heterogeneous uncertain linear MASs with an unknown exosystem via ETC. By
constructing a distributed adaptive internal model, a fully distributed ETC
strategy is proposed for each agent. It is shown that under the proposed ETC
strategy, the tracking errors converge to the origin with any specified level of accuracy whilst the
Zeno behavior is ruled out. In future work, the same problem will be studied
for more general agent dynamics or even nonlinear MASs.


\bibliographystyle{elsarticle-num}
\bibliography{mybib6}

\end{document}